\documentclass[twocolumn,amsthm]{autart}

\usepackage{lmodern}
\usepackage[utf8]{inputenc}

\usepackage{mathtools}
  \mathtoolsset{showonlyrefs, centercolon, mathic}
\usepackage{amssymb}
\usepackage{dsfont}
\usepackage{nicefrac}
\usepackage{siunitx}

\usepackage[usenames, dvipsnames, svgnames]{xcolor}
\usepackage{tikz}
  \usetikzlibrary{arrows}
    \tikzset{>=stealth}
  \usetikzlibrary{backgrounds}
  \usetikzlibrary{external}
    \tikzset{external/only named}
    \tikzset{external/mode=list and make}
    \tikzset{external/prefix=fig/}
    \tikzset{external/system call=%
        {lualatex \tikzexternalcheckshellescape -synctex=1 %
        -interaction=batchmode -jobname "\image" "\texsource"}%
    }
    \pgfkeys{/pgf/images/include external/.code={%
        \IfFileExists{#1.pdf}{%
          \includegraphics{{#1}.pdf}%
        }{%
          \includegraphics[draft]{{#1}.pdf}%
        }%
      }
    }
  \usetikzlibrary{plotmarks}
  \usetikzlibrary{positioning}

\usepackage{pgfplots}
  \usepgfplotslibrary{groupplots}
  \usepgfplotslibrary{statistics}
  \pgfplotsset{compat=newest}
  \pgfplotsset{
    log ticks with fixed point,
    scaled ticks=false,
    clip marker paths=true,
    xlabel near ticks,
    ylabel near ticks,
    every axis title shift=1pt,
  }

\usepackage{booktabs,colortbl}
\usepackage{pgfplotstable}
  \pgfplotstableset{
    every even row/.style={before row={\rowcolor[gray]{0.9}}},
    every head row/.style={before row=\toprule,after row=\midrule},
    every last row/.style={after row=\bottomrule},
  }

\usepackage{enumerate}
\usepackage{dblfloatfix}
\usepackage{paralist}
\usepackage[multidot]{grffile}
\usepackage{needspace}
\usepackage{xifthen}
\usepackage{todonotes}

\usepackage{natbib}

\usepackage[colorlinks=true, linkcolor=black, citecolor=black, 
            urlcolor=blue]{hyperref}

\newcommand{\trans}{^{\mathsf{T}}} 
\newcommand{\mtrans}{^{\mathsf{-T}}} 
\newcommand{\inv}{^{-1}} 

\newcommand{\ident}{\ensuremath{\mathrm{I}}} 
\newcommand{\dd}{\ensuremath{\mathrm{d}}} 
\newcommand{\reals}{\ensuremath{\mathds{R}}} 
\DeclareMathOperator*{\maximize}{maximize} 
\DeclareMathOperator*{\subjto}{subject\;to} 

\newcommand{\pderiv}[2]{\ensuremath{\frac{\partial #1}{\partial #2}}}

\newcommand{\paren}[1]{\ensuremath{\left(#1\right)}}

\newcommand{\fparen}[1]{\ensuremath{\!\left(#1\right)}}

\newcommand{\Prob}[2][]{%
    \ensuremath{P\ifthenelse{\isempty{#1}}{}{_{#1}\!}\fparen{#2}}%
}

\newcommand{\pr}[2][]{%
  \ensuremath{p\ifthenelse{\isempty{#1}}{\!}{_{#1}\negmedspace}\paren{#2}}%
}

\newcommand{\snorm}[1]{%
  \left\lvert\kern-0.3ex\left\lvert\kern-0.3ex\left\lvert
  #1
  \right\rvert\kern-0.3ex\right\rvert\kern-0.3ex\right\rvert
}

\newcommand{\snormnm}[1]{%
  \lvert\kern-0.3ex\lvert\kern-0.3ex\lvert
  #1
  \rvert\kern-0.3ex\rvert\kern-0.3ex\rvert
}

\newcommand{\E}[2][]{\ensuremath{\mathrm{E}_{#1}\negmedspace\left[#2\right]}}

\newcommand{\normald}[1]{\ensuremath{\mathcal{N}\!\paren{#1}}}

\newcommand{\cond}{\ensuremath{\,\middle\vert\,}}

\newcommand{\inputtikzpicture}[1]{%
  \includegraphics{#1.tikz.pdf}%
}

\newcommand{\popt}{p^*}
\newcommand{\qopt}{q^*}
\newcommand{\zopt}{z^*}
\newcommand{\lmopt}{\lambda^*}
\newcommand{\pspace}{{\mathbb{P}}}
\newcommand{\qspace}{{\mathbb{Q}}}
\newcommand{\zspace}{{\mathbb{Z}}}
\newcommand{\pnei}{{\mathbb{P}^*}}
\newcommand{\qnei}{{\mathbb{Q}^*}}

\begin{document}

\begin{frontmatter}
  \title{%
    Uncertainty estimation in equality-constrained MAP and maximum likelihood
    estimation with applications to system identification and state estimation
  }

  \author{Dimas Abreu Archanjo Dutra}\ead{dimasad@ufmg.br}

  \address{%
    Departamento de Engenharia Mecânica ---
    Universidade Federal de Minas Gerais ---
    Belo Horizonte, MG, Brasil
  }
  
  
  \begin{abstract}
    In unconstrained maximum \emph{a posteriori} (MAP) and maximum likelihood
    estimation, the inverse of minus the merit-function Hessian
    matrix is an approximation of the estimate covariance matrix. 
    In the Bayesian context of MAP estimation, it is the covariance of a normal 
    approximation of the posterior around the mode; while in maximum likelihood
    estimation, it an approximation of the inverse Fisher information 
    matrix, to which the covariance of efficient estimators converge.
    These measures are routinely used in system identification to evaluate the
    estimate uncertainties
    and diagnose problems such as overparametrization, improper excitation
    and unidentifiability.
    A wide variety of estimation problems in systems and control, however, can
    be formulated as equality-constrained
    optimizations with additional decision variables to
    exploit parallelism in computer hardware, simplify implementation and
    increase the convergence basin and efficiency of the nonlinear program
    solver.
    The introduction of the extra variables, however, dissociates the 
    inverse Hessian from the covariance matrix.
    Instead, submatrices of the inverse Hessian of the constrained-problem's 
    Lagrangian must be used. 
    In this paper, we derive these relationships, showing how the
    estimates' covariance can be estimated directly from the augmented problem.
    Application examples are shown in system identification with the 
    output-error method and joint state-path and parameter estimation.
  \end{abstract}
\end{frontmatter}

\section{Introduction}
\label{sec:introduction}

When formulating a nonlinear optimization problem, function composition in
the merit function can be replaced by the introduction of
additional decision variables and equality constraints.
This general technique can simplify the implementation,
transform a dense optimization problem into a sparse one, and help overcome
local optima and convergence issues
\citep{dutra2019cbo, ribeiro2019sns}.
In the output-error method for system identification, this principle underlies
the use of collocation and multiple shooting in continuous-time 
[\citealp{bock1980nti, bock1982rap};
\citealp{betts2003lsp}; \citealp{williams2005ope};
\citealp[Chap.~5]{betts2010pmo}; \citealp{boisvert2012cod}; 
\citealp{dutra2019cbo}] and discrete-time systems \citep{ribeiro2017smp}.
The approach is pretty general, however, and has been applied to other state
\citep{lopez2012mhe, dutra2014map} and joint state and parameter estimators 
\citep{dutra2012jmaps, dutra2017jmp}.

In unconstrained maximum \emph{a posteriori} (MAP) and maximum likelihood (ML)
estimation, the inverse of the merit-function Hessian matrix is an 
approximation of the estimates' covariance matrix.
This is used extensively  to obtain estimate correlations, uncertainties, 
and significance in system identification \citep[Chap.~16]{ljung1999si}
 and its applications to aircraft
[\citealp{maine1981tpe}; \citealp[Sec.~6.3.3]{klein2006asi}; 
\citealp[Sec.~11.2]{jategaonkar2015fvs}].
The covariance matrix, merit-function Hessian, and measures derived from them
can, in turn, be used to diagnose overparametrization, unidentifiability,
improper excitation, and inadequate model postulates 
\citep{mehra1974ois, soderstrom1975coa, stoica1982ons}.
From the standpoint of the optimization problem, the Hessian quantifies how
deep or shallow the basin around the optimum is. 
A small eigenvalue means that a certain parameter or parameter combination
changes the posterior density or the likelihood very little and, consequently,
is not well estimated from the test data.

There are also formal interpretations to this intuitive concept that justify
the use of the Hessian matrix to quantify the estimates' uncertainty.
In classical statistics, the parameters are considered
unknown deterministic quantities and their data-dependent estimates are random
variables.
Under some regularity conditions, maximum likelihood estimators are 
asymptotically consistent, unbiased, efficient, and normally distributed
[see, e.g., \citealp[Sec.~33.3]{cramer1946mms}; \citealp{wald1949ncm};
\citealp[Sec.~12.3]{wilks1962ms}].
The regularity conditions are associated with a well-posed ML estimator for
sufficient data, such as having unique global maxima for all large enough
datasets.
The covariance of efficient estimators is given by the inverse of the Fisher
information matrix.
If the estimator is consistent, minus the Hessian of the log-likelihood
at the optima converges to the Fisher information by the law of large numbers.

These conditions apply for a wide class of prediction error methods in system
identification \citep[Chaps.~8--9]{kashyap1970mli, caines1976pee, ljung1999si} 
and are assumed to hold for aircraft system 
identification methods such as the output error method and the filter error
method [\citealp[Sec.~III.B-2]{murphy1985aem}; 
\citealp[Appx.~D]{jategaonkar2015fvs}].
This makes the merit-function Hessian an asymptotic measure of the 
estimate uncertainties for a wide range of problems in systems and control.
Its use in the practice of aircraft system identification is supported by
experimental evidence and simulated analyses \citep[see, e.g.,][]{maine1981tpe}.

Under a Bayesian framework of MAP estimation,
the merit-function is the log-posterior, the logarithm of the posterior 
probability density.
Under some regularity conditions, the Bayesian Central Limit Theorem
\citep[Sec.~3.2.1]{ando2010bms}, also
known as the Bernstein--von Mises Theorem \citep[Sec.~10.2]{vaart1998as}
states that the posterior converges asymptotically to a normal distribution
centered around the mode, with covariance given by minus the inverse 
merit-function Hessian.
This is a Bayesian analog of the limiting behavior of the ML estimator which,
as data grows, also converges to the same value as the posterior mode.

The introduction of aditional variables and equality constraints 
in approaches like collocation and multiple shooting,
however, dissociates the merit-function Hessian from
the Fisher information matrix.
In this work, we prove that the inverse Hessian of the equivalent
unaugmented problem is \emph{equal} to a submatrix of the inverse
Hessian of the Lagrangian of the augmented, equality-constrained, problem.
We also show how to approximate the dependent variables' covariance, based
on a linear approximation of the constraints, from the Lagrangian Hessian.
This allows the computation of the estimate uncertainties directly from the
functions and results used to solve the augmented nonlinear program (NLP).
We note that, although our motivation is system identification and state 
estimation, the
results are general and apply to any equality-constrained maximum likelihood
or maximum \emph{a posteriori} estimator, even beyond the field of systems and
automatic control.

A computationally similar, although conceptually different
solution to the same problem is given by \citet[Sec.~2.6]{ramsay2007ped}, 
\citet[Sec.~4]{pirnay2012osb} and \citet[Sec.~2.3]{lopez2012mhe}.
As explained in detail by \citet[Sec.~7.7]{bard1974npe}, it consists
of evaluating the first-order change in the estimates generated by a 
stochastic perturbation to the data.
The covariance of the estimate variation then gives, approximatelly, the
estimate uncertainties, especially when that variation is small and hence the
linear approximation is reasonable.
In nonlinear least-square problems both approaches give the same result, 
although in non-Gaussian ML problems like those associated with robust 
estimation they may differ.
It should be noted, however, that this approach of considering stochastic
perturbations to the data does not rely on the assumptions of consistency and
efficiency of the estimator, so it is applicable to other optimization-based
methods besides MAP and maximum likelihood \citep[Sec.~7.5]{bard1974npe}.

The remainder of this article is organized as follows.
In Sec.~\ref{sec:results} we define the reduced and augmented estimation
problems and present the theoretical result.
In Sec.~\ref{sec:app} we present applications of the result to the output-error
method of system identification, joint MAP state-path and parameter estimation
in stochastic differential equations (SDEs), and Markov chain Monte Carlo
(MCMC) sampling of state-paths and parameters in SDEs.
The paper concludes with the final remarks in Sec.~\ref{sec:conclusions}.

\section{Theoretical result}
\label{sec:results}

\subsection{Notation}

For a scalar function $f\colon\reals^n\times\reals^m\to\reals$ of two 
vector-valued
arguments, $f(x,y)$, we denote its gradient with respect to $x$ by 
$\nabla_x f(a,b)\in\reals^{n}$ and Hessian matrix with respect to $x$ and $y$ by
$\nabla^2_{xy} f(a,b)\in\reals^{n\times m}$:
\begin{align}
  [\nabla_x f(a,b)]_{i} &= 
  \left.\pderiv{f}{x_i}\right|_{\substack{x=a\\y=b}}, \\
  [\nabla_{xy}^2 f(a,b)]_{ij} &= 
  \left.\pderiv{^2 f}{x_i\partial y_j}\right|_{\substack{x=a\\y=b}}.
\end{align}
Similarly, for vector-valued functions 
$g\colon\reals^n\times\reals^m\to\reals^p$
of two vector-valued variables $g(x,y)$, we define the Jacobian matrix with 
respect to $x$ as $\nabla_x g(a,b)\in\reals^{p\times n}$:
\begin{align}
  [\nabla_x g(a,b)]_{ij} = 
  \left.\pderiv{g_i}{x_j}\right|_{\substack{x=a\\y=b}}.
\end{align}
For functions of only one argument the subscript to the nabla symbol may be 
omitted.
In addition, $\nabla_{xx}^2$ will be shortened to $\nabla_x^2$.

\subsection{Augmented and reduced problem definitions}

We now begin with the definition of the estimation problems under consideration.
Our starting point is the augmented problem, written in terms of the independent
variables $p\in\pspace\subseteq\reals^n$ and the dependent variables 
$q\in\qspace\subseteq\reals^m$:
\begin{equation}
  \label{eq:aug_problem}
  \begin{aligned}
    \operatorname*{maximize}_{p\in\pspace, q\in\qspace} &&\;&
    \ell(p,q)\\
    \operatorname{subject\,to}\;&&&
    g(p,q)=0,
  \end{aligned}
\end{equation}
where $\ell\colon\pspace\times\qspace\to\reals$ is the merit function and the
equality constraints are encoded by $g\colon\pspace\times\qspace\to\reals^m$.
To simplify notation, we let $z\in\zspace:=\pspace\times\qspace$ 
denote the full decision variable vector of the augmented problem, i.e., 
$z = [p\trans\; q\trans]\trans$ and let $\ell$ and $g$ be written in
terms of either $z$ or the $(p,q)$ pair.

We consider only well-posed problems which can be rewritten as unconstrained 
optimizations only in terms of 
the independent variables $p$, at least around the optimum.
Formally, this translates to the assumptions below.

\begin{assum}
  \label{th:assum}
  \hfill
  \begin{enumerate}[a.]
  \item \label{it:existence_sol}%
    There exists a unique solution $\popt,\qopt$ 
    to the problem defined in Eq.~\eqref{eq:aug_problem}, and it lies in 
    the interior of the search space, $\popt\in\pspace^\circ$, 
    $\qopt\in\qspace^\circ$.
  \item \label{it:differentiability}%
    The merit $\ell$ and constraint $g$ are twice continuously differentiable
    (of class $C^2$) with respect to all arguments at a neighbourhood
    $\pnei\times\qnei\subseteq\pspace\times\qspace$ of the optimum
    $\zopt= [\popt{}\trans\;\qopt{}\trans]\trans$.
  \item \label{it:jac_inv}%
    The Jacobian matrix $\nabla_q g(\zopt)$ of $g$ with respect to 
    the dependent variables $q$, evaluated at the optimum, is invertible.
  \end{enumerate}
\end{assum}

Assum.~\ref{th:assum} implies that the implicit function theorem 
can be used, consequently there
exists a unique explicit solution to the constraints around the optimum.

\begin{cor}[{of the implicit and inverse function theorems, 
    \citealp[Thms.~3.3.1 and 3.3.2]{krantz2003ift}}]
  \label{th:explicit_sol}
  If Assum.~\ref{th:assum} holds then, for a neighbourhood 
  $\pnei\subseteq\pspace$ of $\popt$,
  there exists a unique twice continuously differentiable (of class $C^2$)
  function $w\colon\pnei\to\reals^{n+m}$ such that $w(\popt) = \zopt$ and
  \begin{equation}
    \label{eq:explicit_sol}
    g\big(w(p)\big) = 0 \qquad \forall p \in \pnei.
  \end{equation}
  Furthermore, the Jacobian matrix of $w$ at the optimum is
  \begin{equation}
    \label{eq:w_jac}
    \nabla w(\popt) = 
    \begin{bmatrix}
      \ident_n \\ 
      -\nabla_q g \inv \nabla_p g
    \end{bmatrix}_{z=\zopt}.
  \end{equation}
\end{cor}

A consequence of Cor.~\ref{th:explicit_sol} is that, at least in the
neighborhood of the solution, the optimization problem can be rewritten
in a reduced form without constraints, in terms of only
the independent variables $p$:
\begin{equation}
  \label{eq:cond_problem}
  \begin{aligned}
    \operatorname*{maximize}_{p\in\pnei} &&\;&
    \tilde\ell(p),
  \end{aligned}
\end{equation}
where the reduced merit $\tilde\ell\colon\pnei\to\reals$ is of class
$C^2$ and given by
\begin{equation}
  \label{eq:cond_merit_defn}
  \tilde\ell(p) \coloneqq \ell\big(w(p)\big) \qquad \forall p\in\pnei.
\end{equation}
Both problems are equivalent in terms of having the same solution.

\subsection{Definition of the Lagrangian and bordered Hessian}

We now define the Lagrangian 
$L\colon\zspace\times\reals^m\to\reals$ of the augmented 
equality-constrained problem \eqref{eq:aug_problem}
so that the theorem can be stated:
\begin{equation}
  \label{eq:lagrangian_defn}
  L(z,\lambda) \coloneqq \ell(z) + g(z)\trans\lambda,
\end{equation}
where $\lambda\in\reals^m$ are the Lagrange multipliers.
A standard result of constrained optimization that follows from 
Assum.~\ref{th:assum}
is that a necessary condition for optimality is the existence of $\lmopt$ such 
that
\begin{align}
  \label{eq:fonc}
  \nabla_z L(\zopt,\lmopt) = \nabla \ell(\zopt) + 
  \nabla g(\zopt)\trans\lmopt = 0.
\end{align}

The Hessian matrix $H_L$ of the Lagrangian with respect to $z$ and $\lambda$ at 
the optimum is then given by
\begin{equation}
  \label{eq:HL_defn}
  H_L \coloneqq 
  \begin{bmatrix}
    \nabla^2 \ell + \sum_{i=1}^{m}\lmopt_i\nabla^2 g_i &\phantom{AA} &
    \nabla g\trans \\
    \nabla g && 0
  \end{bmatrix}_{z=\zopt}.
\end{equation}
The $H_L$ matrix is sometimes called the bordered Hessian and is used in the
numerical solution of the optimization problem.
It is usually readily available to the user for evaluation after successful 
termination of the optimization.

\subsection{Equivalence of inverse Hessians}
We are now ready to present the main result of this article.

\begin{thm}
  \label{th:main}
  If Assum.~\ref{th:assum} holds, then both $\nabla^2\tilde\ell(\popt)$ 
  and $H_L$ are invertible.
  Furthermore, the inverse of the Hessian matrix of
  the reduced problem \eqref{eq:cond_problem} at the optimum
  $[\nabla^2 \tilde \ell(\popt)]\inv$
  is given by the first $n$ columns and $n$ rows of $H_L\inv$.
\end{thm}

To prove Thm.~\ref{th:main} we obtain the expression
of the inverse reduced Hessian and show that it is
equal to a submatrix of the blockwise inversion of $H_L$.
To obtain the reduced Hessian, we state a lemma often used
to obtain sufficient second-order conditions for constrained optima.
To simplify notation, in what follows all derivatives are evaluated at
$p=\popt$ and $z=\zopt$.

\begin{lem}
  \label{th:cond_hess_expr}
  The Hessian matrix of the reduced problem \eqref{eq:cond_problem},
  at the optimum, is given by
  \begin{equation}
    \label{eq:cond_hess}
    \nabla^2\tilde\ell =
    \nabla w \trans
    \left[
      \nabla^2 \ell + \sum_{s=1}^{m}\lmopt_s\nabla^2 g_s
    \right]
    \nabla w.
  \end{equation}
\end{lem}
\begin{proof}
  Using the chain rule, we have that
  \begin{gather}
    \pderiv{\tilde\ell}{p_i} =
    \sum_{k=1}^{m+n} \pderiv{\ell}{z_k} \pderiv{w_k}{p_i} \\
    \label{eq:cond_hess_elems}
    \pderiv{^2\tilde\ell}{p_i\partial p_j} =
    \sum_{k=1}^{m+n}\sum_{r=1}^{m+n} 
    \pderiv{^2\ell}{z_k\partial z_r} \pderiv{w_k}{p_i} \pderiv{w_r}{p_j}
    + \sum_{k=1}^{m+n} \pderiv{\ell}{z_k}\pderiv{^2 w_k}{p_i\partial p_j}.
  \end{gather}
  Furthermore, from the first-order necessary conditions for the optimum 
  \eqref{eq:fonc}
  we have that
  \begin{equation}
    \label{eq:unconstrained_grad_mult}
    \pderiv{\ell}{z_k} = -\sum_{s=1}^{m}\lmopt_s\pderiv{g_s}{z_k},
  \end{equation}
  leading to the following expression for the last term of 
  Eq.~\eqref{eq:cond_hess_elems}:
  \begin{equation}
    \sum_{k=1}^{m+n} \pderiv{\ell}{z_k}\pderiv{^2 w_k}{p_i\partial p_j} = 
    -\sum_{s=1}^{m}\lmopt_s
    \sum_{k=1}^{m+n} \pderiv{g_s}{z_k}\pderiv{^2 w_k}{p_i\partial p_j}
  \end{equation}

  Next, define $\tilde g(p):=g\big(w(p)\big)$.
  Applying the chain rule,
  \begin{gather}
    \pderiv{\tilde g_s}{p_i} = 
    \sum_{k=1}^{n+m}\pderiv{g_s}{z_k}\pderiv{w_k}{p_i},
    \\
    \pderiv{^2 \tilde g_s}{p_i\partial p_j} = 
    \sum_{k=1}^{n+m}\sum_{r=1}^{n+m}
    \pderiv{^2 g_s}{z_k\partial z_r}\pderiv{w_k}{p_i}\pderiv{w_r}{p_j} +
    \sum_{k=1}^{n+m}
    \pderiv{g_s}{z_k}\pderiv{^2 w_k}{p_i\partial p_j}.
  \end{gather}  
  As $w(p)$ is an explicit solution to the constraints, $\tilde g$ and all its 
  derivatives are identically zero for all $p\in\pnei$, so
  \begin{equation}
    \label{eq:constraint_hessian}
    \sum_{k=1}^{n+m}
    \pderiv{g_s}{z_k}\pderiv{^2 w_k}{p_i\partial p_j} = 
    -\sum_{k=1}^{n+m}\sum_{r=1}^{n+m}
    \pderiv{^2 g_s}{z_k\partial z_r}\pderiv{w_k}{p_i}\pderiv{w_r}{p_j}.
  \end{equation}
  Substituting \eqref{eq:constraint_hessian} and 
  \eqref{eq:unconstrained_grad_mult}
  into \eqref{eq:cond_hess_elems}, we obtain
  \begin{multline}
    \pderiv{^2\tilde\ell}{p_i\partial p_j} =\\
    \sum_{k=1}^{m+n}\sum_{r=1}^{m+n} 
    \paren{
      \pderiv{^2\ell}{z_k\partial z_r} + 
      \sum_{s=1}^{m}\lmopt_s \pderiv{^2 g_s}{z_k\partial z_r}
    }
    \pderiv{w_k}{p_i}\pderiv{w_r}{p_j},
  \end{multline}
  which is the elementwise expression of \eqref{eq:cond_hess}.
\end{proof}

\begin{cor}
  \label{th:H_cond_inv}
  If Assum.~\ref{th:assum} holds, then $\nabla^2\tilde\ell$ is invertible.
\end{cor}
\begin{proof}
  Note that $\nabla w$ is a basis for the null space of $\nabla g$, which
  can be verified by taking their product.
  Lemma~\ref{th:cond_hess_expr} then amounts to the second-order 
  sufficient and necessary conditions for constrained optimality
  \citep[Thms.~12.5 and 12.6]{nocedal2006no}.
  Additionally, as $\zopt$ is assumed to be the unique constrained maximum,
  a consequence of the necessary conditions is that $\nabla^2\tilde\ell$
  is negative-definite, hence invertible.
\end{proof}

We are now ready to prove the main theorem.

\begin{figure*}[t]
  \begin{equation}
    \label{eq:HL_inv}
    H_L\inv = 
    \begin{bmatrix}
      (A-B D\inv B\trans)\inv & 
      -(A-BD\inv B\trans)\inv BD\inv \\
      -D\inv B\trans (A-BD\inv B\trans)\inv &\qquad
      D\inv +D\inv B\trans (A-BD\inv B\trans)\inv BD\inv
    \end{bmatrix}
  \end{equation}
  \hrulefill
\end{figure*}

\begin{proof}[Proof of Thm.~\ref{th:main}]
  To begin, note that from the definition of $L$ in \eqref{eq:lagrangian_defn}, 
  we have that the bracketed expression in \eqref{eq:cond_hess} is 
  $\nabla^2_z L$, which has the following block structure:
  \begin{equation}
    \label{eq:lag_hess_z}
    \nabla_z^2 L = \nabla^2 \ell + \sum_{i=1}^{m}\lmopt_i\nabla^2 g_i = 
    \begin{bmatrix}
      \nabla_p^2 L & \nabla_{qp}^2 L\trans \\
      \nabla_{qp}^2 L & \nabla_q^2 L
    \end{bmatrix}.
  \end{equation}
  Substituting \eqref{eq:w_jac} and \eqref{eq:lag_hess_z} into 
  \eqref{eq:cond_hess}, 
  \begin{multline}
    \label{eq:cond_hess_terms}
    \nabla^2 \tilde \ell = 
    \nabla_p^2 L - \nabla_{qp}^2 L\trans \nabla_q\inv\nabla_p g - 
    \nabla_p g\trans\nabla_qg\mtrans\nabla_{qp}^2 L \\
    + \nabla_p g\trans\nabla_q g\mtrans\nabla_q^2 L \nabla_q g\inv\nabla_p g.
  \end{multline}
  
  Next we obtain the expression of the first $n$ rows and $n$ columns of 
  $H_L\inv$.
  The $H_L$ matrix of \eqref{eq:HL_defn} has the following block structure:
  \begin{subequations}
    \label{eq:HL_block}
    \begin{align}
      \label{eq:HL_block_a}
      H_L &=
      \begin{bmatrix}
        A & B \\ B\trans & D
      \end{bmatrix}, &
      D &:= 
      \begin{bmatrix}
        \nabla_q^2 L & \nabla_q g\trans \\
        \nabla_q g & 0
      \end{bmatrix}, \\
      \label{eq:HL_block_b}
      A &:= \nabla_p ^2 L, & 
      B &:=
      \begin{bmatrix}
        \nabla_{qp}^2 L\trans & \nabla_p g\trans
      \end{bmatrix}.
    \end{align}
  \end{subequations}%
  \noeqref{eq:HL_block_a, eq:HL_block_b}%
  From the blockwise inversion formula \cite[Prop.~2.8.7]{bernstein2009mm},
  we have that if $D$ and $A-B D\inv B\trans$ are nonsingular,
  then so is $H_L$ and its inverse is given by \eqref{eq:HL_inv}.
  
  As $\nabla_q g$ is invertible due to Assum.~\ref{th:assum}\ref{it:jac_inv},
  then $D$ is invertible,
  \begin{equation}
    \label{eq:lower_block_inverse}
    D\inv = 
    \begin{bmatrix}
      0 & \nabla_q g\inv \\
      \nabla_q g \mtrans & -\nabla_q g\mtrans \nabla_q^2 L \nabla_q g\inv
    \end{bmatrix},
  \end{equation}
  which can be verified directly by taking the product.
  Again, taking the product and comparing to the right-hand side of 
  \eqref{eq:cond_hess_terms},
  \begin{equation}
    \label{eq:H_cond_inv_blk}
    A-B D\inv B\trans = \nabla^2\tilde \ell,
  \end{equation}
  which by Cor.~\ref{th:H_cond_inv} is invertible.
  Consequently, the blockwise inversion formula is applicable and from 
  \eqref{eq:HL_inv} and \eqref{eq:H_cond_inv_blk} we have that 
  $(\nabla^2\tilde \ell)\inv$ equals the first $n$ rows and $n$ columns of
  $H_L\inv$.
\end{proof}

\subsection{Approximate covariance of all decision variables}

When the inverse Hessian is used as an approximation of the independent
variables' covariance, a natural question that arises is if it can be used
to obtain the covariance of the full vector of decision variables, 
or at least an estimate of.
If a linear approximation of the constraints and, consequently, of the 
relationship between the independent and dependent variables $q$ and $p$ is
used, 
\begin{equation}
  w(p) - \zopt \approx (\nabla w) (p - \popt)
\end{equation}
for $p$ in a neighbourhood of the optimum $\popt$.
This approximation, together with the use of $(\nabla^2\tilde \ell)\inv$ for
the covariance of the estimates, yields
\begin{subequations}
  \label{eq:aug_dec_cov}
  \begin{align}
    \label{eq:aug_dec_cov_classical}
    \E{(\zopt - \bar z)(\zopt - \bar z)\trans} &\approx
    (\nabla w) (\nabla^2\tilde \ell)\inv (\nabla w)\trans,\\
    \label{eq:aug_dec_cov_bayesian}
    \E{(z-\zopt)(z-\zopt)\trans \cond y} &\approx
    (\nabla w) (\nabla^2\tilde \ell)\inv (\nabla w)\trans,
    \qquad\;
  \end{align}
\end{subequations}%
\noeqref{eq:aug_dec_cov_classical, eq:aug_dec_cov_bayesian}%
in which $\bar z$ is the true value of the estimates, in the context of
maximum likelihood estimation, and the expected value of $\zopt$ if the
estimator is unbiased.
In a Bayesian context \eqref{eq:aug_dec_cov_bayesian}, $z$ is the random 
variable
being estimated, $y$ is the data, and the posterior mode $\zopt$ is used as the mean.

The results of the previous section can be used to obtain an expression of
the right-hand sides of \eqref{eq:aug_dec_cov} in terms of $H_L\inv$.
These are summarized in the following Lemma.

\begin{lem}
  \label{th:aug_dec_cov}
  If Assum.~\ref{th:assum} holds, then 
  $(\nabla w) (\nabla^2\tilde \ell)\inv (\nabla w)\trans$
  is given by the first $m+n$ rows and $n+m$ columns of $H_L\inv$.
\end{lem}
\begin{proof}
  Using the results obtained previously, a direct proof is possible.
  By substituting \eqref{eq:HL_block} and \eqref{eq:lower_block_inverse} into 
  \eqref{eq:HL_inv}, we can obtain the expression of $H_L\inv$.
  Similarly, \eqref{eq:w_jac} and \eqref{eq:cond_hess_terms} can be substituted
  into $(\nabla w) (\nabla^2\tilde \ell)\inv (\nabla w)\trans$ and we can see
  that the statement holds.
\end{proof}

It should be noted that the use of the approximation 
\eqref{eq:aug_dec_cov_classical} in the context of classical statistics and
maximum likelihood estimation is better justified when the constraints are not
data-dependent.
In this case, the same function $w$ of Cor.~\ref{th:explicit_sol} is valid for
all realizations of the experiment.
This occurs in the output-error method, for example, as the constraints encode 
the solution of a deterministic ordinary differential equation or 
difference equation.

An important issue in using these methods is that although in most large 
problems of practical interest the Lagrangian Hessian is sparse, its 
inverse is generally dense.
A workaround to this, which we use in the example of Sec.~\ref{sec:oem}, is to 
obtain individual columns of $H_L\inv$ by solving the linear system
$H_Lx=b$, in which $b$ is a column of the identity matrix.
This allows some unused elements of the covariance matrix to be discarded,
keeping the memory requirements low, the same solution used by
\citet[Sec.~2.2]{lopez2012mhe}.
We also remark that the results presented in this section cannot be used
to diagnose unidentifiability, as then Assum.~\ref{th:assum} would not
hold.
Correlation coefficients close to \num{+-1} can, however, be 
indicative of poorly identifiable, ill-conditioned problems, often
associated with overparametrization, improper excitation, or inadequate model
postulates.

\section{Application examples}
\label{sec:app}
We now consider applications of the results of Sec.~\ref{sec:results} to 
simulated problems in systems and control.
The first example consists of the estimation of the parameter and state-path
covariance
in the system identification of a nonlinear continuous-time system using the
collocation-based output-error method, formulated as a maximum likelihood 
problem.
The second example is an application to joint state-path and parameter 
estimation in stochastic differential equations (SDEs).
The estimated posterior covariance is used to calibrate a Markov chain 
Monte Carlo (MCMC) sampler of state-paths and parameters.

The implementation of the estimators used herein is in the open-source
software package \texttt{ceacoest}%
\footnote{Available at \url{http://github.com/cea-ufmg/ceacoest}},
the \emph{Centro de Estudos Aeronáuticos} Control and Estimation library,
under development by the author.
It has bindings for the large-scale nonlinear optimization solver IPOPT by
\citet{wachter2006iip}, which was used together with the HSL Mathematical
Software Library%
\footnote{%
  HSL. A collection of Fortran codes for large scale scientific computation. 
  \url{http://www.hsl.rl.ac.uk/}
}.
The code used to generate the data and analyses for this article is also 
available as free software%
\footnote{\url{http://github.com/dimasad/hessinv-code}}.

\subsection{Collocation-based output-error estimation}
\label{sec:oem}

The output-error method is a standard method for system identification and
parameter estimation.
Given the system parameters and initial condition, the system is simulated and 
the output error minimized, according to a loss function associated to a 
measurement noise distribution.
When there is no prior distribution for the estimates it is a maximum 
likelihood estimator.

When implemented with collocation methods, the decision variables of the 
optimization problem are augmented with the state vector at all simulation 
points and the simulation method is enforced as as equality constraints 
\citep[for a brief overview, see][]{dutra2019cbo}.
We used the trapezoidal rule for the collocation in these experiments,
see \citet{williams2005ope} for this and many alternative methods for 
collocation.
\citet{betts2003lsp} and \citet{betts2010pmo} discuss implementation issues in 
detail.

To illustrate the estimation of uncertainties, we perform simulated experiments
on the Van der Pol oscillator, a benchmark model
for modelling nonlinear dynamics and chaos \citep[Appx.~A.2]{aguirre2009mnd}.
Its dynamics is governed by the ODE
\begin{align}
  \label{eq:vdp}
  \dot x_1 &= x_2, &
  \dot x_2 &=\mu\big(1- x_1^2\big) x_2 - x_1,
\end{align}
where $x=[x_1, x_2]\trans$ is the state vector and
$\theta=[\mu\;\sigma]\trans$ is the unknown parameter vector, to be estimated.
Noisy measurements of the first state $x_1$ are available with variance
$\sigma^2$.
The system is simulated for the interval $t\in[0, 20]$, starting from 
$x_1(0)=0$ and $x_2(0)=1$.
The true values of the parameters, used to generate the data, are $\mu=2$ and 
$\sigma=\num{0.1}$.
The measurements are spaced by \num{0.1} time units and the collocation mesh
spacing is \num{0.01} time units.

\begin{figure}[b]
  \centering
  \inputtikzpicture{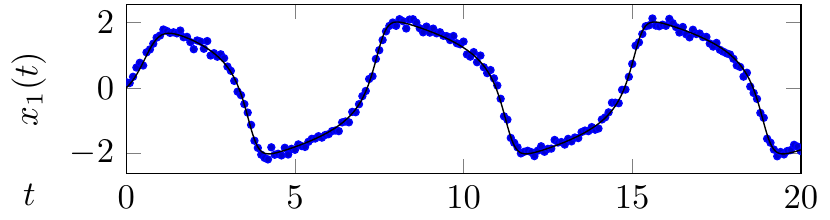}%
  \caption{
    Data from one realization of the Van der Pol oscillator.
    The true state path is the solid line and the marks are the noisy
    measurements.
  }
  \label{fig:example}
\end{figure}

\begin{figure}[b]
  \centering
  \inputtikzpicture{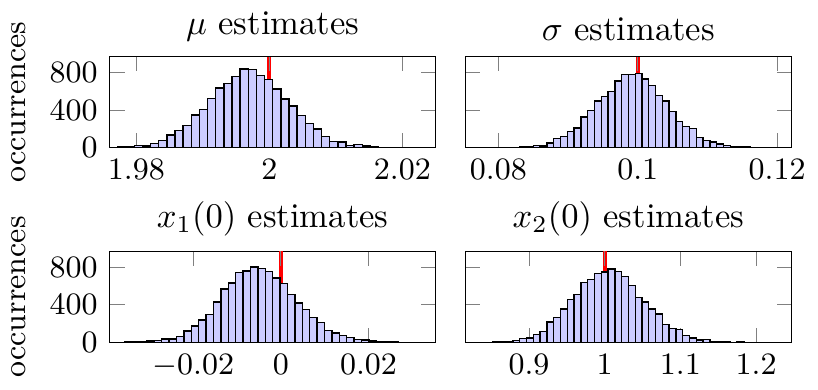}
  \caption{
    Histogram of the Van der Pol parameter estimates. 
    The red line in the background marks the true values, used to generate 
    the data.
  }
  \label{fig:hist}
\end{figure}

\begin{figure}
  \centering
  \inputtikzpicture{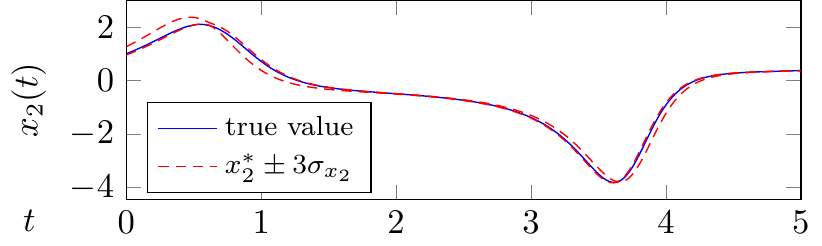}%
  \caption{
    Estimated $x_2$ path confidence bounds compared with the true simulated 
    values for a realization of the Van der Pol experiment.
  }
  \label{fig:state_path_uncertainties}
\end{figure}

\begin{table}[t]
  \caption{
    Estimation results for the simulated experiment.
    The first column is the sample standard deviation of the estimate; 
    the second is the mean estimated standard deviation $\hat\Sigma$ of each 
    estimate, obtained from the inverse Hessians; the third is the sample 
    standard deviation of $\hat\Sigma$.
  }
  \label{tab:res_par}
  \centering
  \begin{filecontents*}{tbl.txt}
{$\mu$}, 0.006021214950691535, 0.005990459838678667, 0.0003059251433387364
{$\sigma$}, 0.005022079450744469, 0.0049479969897423, 0.00025047855694188936
{$x_1(0)$}, 0.008519619149611656, 0.008523143460253102, 0.0007138412322596689
{$x_2(0)$}, 0.04632489792682695, 0.04593189665830923, 0.003036609317084212
\end{filecontents*}
\pgfplotstabletypeset[
    header=false, col sep=comma,
    columns/0/.style={string type, column name={Estimate}},
    columns/1/.style={
      fixed, fixed zerofill, precision=4, column name={std.\ dev.}
    },
    columns/2/.style={
      fixed, fixed zerofill, precision=4, column name={mean $\hat\Sigma$}
    },
    columns/3/.style={
      fixed, fixed zerofill, precision=4, column name={std.\ dev.\ $\hat\Sigma$}
    },
  ]{tbl.txt}
\end{table}

The output error method needs an initial guess for all decision variables.
The initial guess for $x_1$ was obtained by low-pass filtering the measurements.
The guess for $x_2$ was the finite difference derivative of the $x_1$ guess.
Finally, the guess of $\mu$ is obtained from the other guesses using linear
least squares.

To evaluate the distribution of the estimates and how it relates to the 
uncertainty estimates obtained from the inverse Hessian, a total of \num{10000} 
different realizations of the noise were performed, one of which is shown
in Fig.~\ref{fig:example}.
The estimated variance $\hat V$ of each estimate is the corresponding element
of the diagonal of $-H_L\inv$.
The estimated standard deviation $\hat\Sigma$ is its square root.
These estimated uncertainties depend on the noise, so their mean and 
scatter is also analyzed.
The results are summarized in Table~\ref{tab:res_par} for the parameter vector
$\theta$ and initial state vector $x(0)$.
A good agreement between the observed and estimated standard deviations was
obtained.
The histograms of the parameter estimates are shown in Fig.~\ref{fig:hist}.
By employing the approximation \eqref{eq:aug_dec_cov_classical} and the result
of
Lemma~\ref{th:aug_dec_cov}, we can also obtain the estimate uncertainty of the
whole state path, shown in Fig.~\ref{fig:state_path_uncertainties}.

\subsection{Joint MAP state-path and parameter estimation}
For general nonlinear dynamical systems subject to noise, the posterior
distribution does not assume tractable closed-form solutions.
In such cases, Monte Carlo methods are a popular and powerful choice
for evaluating the posterior in detail, including features such as
skewness, excess kurtosis and nonlinear relationships between the variables.
To that end, estimates of the posterior covariance matrix obtained
from the results of
Sec.~\ref{sec:results} can aid the implementation of efficient Monte Carlo
samplers.
In Markov chain Monte Carlo with the random-walk Metropolis algorithm,
for example,
the target distribution's covariance can be used to tune the jumping
scales and geometry for efficient sampling \citep{gelman1996emj}.

We demonstrate this application in a nonlinear continuous-time
system described by a stochastic differential equation (SDE), the Duffing
oscillator. 
It is a benchmark model for modeling 
nonlinear dynamics and chaos \citep{aguirre2009mnd} and state estimation in
SDEs \citep{ghosh2008sis, khalil2009nfc, namdeo2007nsd}.
The example and the estimator are similar to those of a previous work
\citep[Sec.~4.1]{dutra2017jmp}.
The dynamics is described by the following stochastic differential equation
\begin{align}
  \dd X_t &= [-A Z_t^3 -BZ_t -DX_t + \gamma\cos t]\dd t + \sigma_d \dd W_t,\\
  \dd Z_t &= X_t\,\dd t,
\end{align}
where $X_t$ and $Z_t$ are the system states; $W_t$ is a Wiener process;
$A$, $B$, and $D$ are unknown parameters, to be estimated; and
$\gamma$ and $\sigma_d$ are known parameters.

Discrete-time measurements $Y_k$ of $Z_t$ with independent Gaussian noise
were used for the estimation,
\begin{equation}
  \label{eq:gauss_meas}
  Y_k |X,Z,\Theta \sim \normald{Z_{k t_s}, \Sigma_y^2},
  \qquad k=0,\dotsc,N,
\end{equation}
where $t_s$ is the sampling period, the $\reals$-valued standard deviation
$\Sigma_y$ is
an unknown parameter, to be estimated, and $\Theta$ is the full vector of
unknown parameters.
Uniform priors for all parameters and initial conditions were used.
The system was simulated using the strong explicit order 1.5 scheme 
\citep[Sec.~11.2]{kloeden1992nss} with a time step of \num{0.005}.
The parameters and initial states used to generate the data are shown in 
Table~\ref{tab:duffing}.

The joint MAP state-path and parameter estimator \citep{dutra2017jmp}
is the solution to the following  optimal control problem:
\begin{equation}
  \label{eq:jme_problem}
  \begin{aligned}
    \maximize & \,&&\ell(x, z, \theta, \eta) \\
    \subjto &&& \dot x = -az^3 -bz -dx + \gamma\cos 
    + \sigma_d\eta \\
    &&&\dot z = x,
  \end{aligned}
\end{equation}
with merit function
\begin{multline}
  \label{eq:jme_merit}
  \textstyle
  \ell(x,z,\theta,\eta) = 
  -\frac12\sum_{k=0}^N \frac{[y_k - z(kt_s)]^2}{\sigma_y}
  \\ \textstyle
  -(N+1)\ln\sigma_y
  -\frac12\int_0^T \eta(t)^2\,\dd t + \frac{Td}{2},
\end{multline}
where $x\colon[0,T]\to\reals$ and $z\colon[0,T]\to\reals$ are the candidate
modal state-paths, $\eta$ is the associated process noise path, and 
$\theta=[a\;b\;d\;\sigma_y]\trans$ is the unknown
parameter vector.
The infinite-dimensional problem \eqref{eq:jme_problem} needs to be discretized
into finite-dimensional NLP for solution.
To do that, we used the same collocation method used to implement the
output-error in Sec.~\ref{sec:oem}, the trapezoidal method.
For some alternatives discretizations and implementation details, see 
\citet{betts2010pmo}.

\begin{table}
  \centering
  \caption{Parameter values used to simulate the Duffing oscillator.}
  \label{tab:duffing}
  \begin{tabular}{cccccccccccc}
    \hline
    $X_0$ & $Z_0$ & $A$ & $B$ & $D$ & $\Sigma_y$ & $\gamma$ & $\sigma_d$
    & $t_s$ & $T$ \\
    1.0 & 1.0 & 1.0 & -1.0 & 0.2 & 0.1 & 0.3 & 0.1 & 0.1 & 200 \\
    \hline
  \end{tabular}
\end{table}

From the approximate covariance matrix, obtained from $H_L\inv$,
we can inspect the posterior correlation coefficients,
shown in Fig.~\ref{fig:duffing_corr}.
These show important aspects of the problem: the states are highly
correlated with their neighbours, across a time window; and the
parameters are correlated with the states across all time instants.
These dependencies must be taken into account in the Metropolis--Hastings 
algorithm for good acceptance rates to be obtained 
\citep[p.~327]{robert2004mcs}.

\begin{figure}[t]
  \centering
  \inputtikzpicture{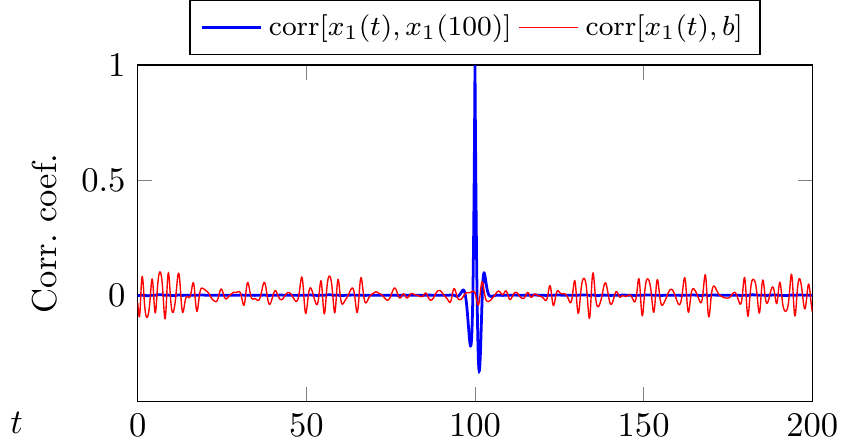}%
  \caption{
    Approximate posterior correlation coefficients between the state-path
    $x_1(t)$ and the $x_1(100)$ state and the $b$ parameter, 
    for the Duffing oscillator.
  }
  \label{fig:duffing_corr}
\end{figure}

A random perturbation of the modal estimate was used to initialize a hybrid
Gibbs sampler \citep[Sec.~10.3]{robert2004mcs} for the posterior distribution
of the state-path and parameters.
We used Gibbs sampling to explore the local features of the posterior 
with a full random-walk Metropolis step after each 15 Gibbs cycles
to increase diversity and better explore the global support of the distribution.
Since the conditional distribution of each variable does not admit a tractable
form, Metropolis-within-Gibbs \citep[Sec.~10.3.3]{robert2004mcs} was used 
to sample the Gibbs steps.

In this problem, $\theta$, $z(0)$ and $x(t_0),\dotsc,x(t_N)$ are
a possible choice for the independent variables $p\in\reals^n$.
For the $i$-th independent variable, given its value at the $j$-th step of the
chain $p_i^{(j)}$, the Metropolis-within-Gibbs candidate was generated as
\begin{align}
  \tilde p_i^{(j+1)} &= p_i^{(j)} + \num{3.2}\epsilon_{ij}, &
  \epsilon_{ij} &\sim \normald{0, R_{ii}},
\end{align}
where the approximate covariance $R:= - (\nabla^2\tilde \ell)\inv$ is the 
inverse reduced Hessian, obtained from Thm.~\ref{th:main}. 
For the full Metropolis step, the candidate was generated according to
\begin{align}
  \tilde p^{(j+1)} &= p^{(j)} + \tfrac{\num{1.8}}{\sqrt n}\epsilon_{j}, &
  \epsilon_{j} &\sim \normald{0, R}.
\end{align}
These candidates were then accepted or rejected with the Metropolis
algorithm acceptance probability.
The scale factors were tuned starting with the values recommended by
\citet{gelman1996emj} and yielded an average acceptance rate of
\SI{26.8}{\%}, close to the \SI{23.4}{\%} suggested by \citet{roberts1997wco}.

\section{Conclusions}
\label{sec:conclusions}

In this article, we showed how to approximate estimate uncertainties in
equality-constrained MAP and maximum likelihood estimation.
These estimators have various applications in system identification and
state estimation, and methods which allow uncertainty estimation directly from
the augmented problem can help in their adoption by a wider userbase.
Two example applications of the results in systems and control were presented,
covering both MAP and ML estimation.

\bibliographystyle{plainnat}
\bibliography{bibtex-compressed}

\end{document}